\documentclass[a4paper, 10pt]{article}
\pdfoutput=1

\usepackage{amssymb}
\usepackage{graphicx}

\usepackage{algorithm}
\usepackage[noend]{algcompatible}
\algrenewcommand\algorithmicindent{2.0em}%
\usepackage{mathtools}
\usepackage{subfig}
\usepackage{tikz}
\usepackage{pgfplots}

\usepackage[version=4]{mhchem}
\tikzstyle{bond2}=[double distance=2pt, thick]
\tikzstyle{bond1}=[thick]
\tikzstyle{bond0}=[color=white, text=black]
\tikzstyle{change}=[style=dotted, thick, text=magenta]
\tikzstyle{dummy}=[inner sep=0pt]

\def\altcyc{\texttt{AltCyc}}
\def\ilptwo{\texttt{ILP2}}
\def\vftwo{\texttt{VF2}}

\def\dom{\mathrm{dom}}
\def\ran{\mathrm{range}}

\usepackage{amsthm}

\newtheorem{proposition}[]{Proposition}

\newtheorem{definition}{Definition}[]

\usepackage{authblk}
\usepackage{marvosym}
\newcommand\corr{\text{ (\Letter)}}
\newcommand\email[1]{\texttt{#1}}

\title{Automatic~Inference~of Graph~Transformation~Rules Using~the~Cyclic~Nature~of~Chemical~Reactions}

\author[2,8]{Christoph Flamm}
\author[1]{ Daniel Merkle\corr}
\author[2-7]{ Peter F.\ Stadler}
\author[1]{ Uffe Thorsen}

\affil[1]{
	Department of Mathematics and Computer Science, University of Southern Denmark, Odense M DK-5230, Denmark\\
	\email{\{daniel, uthorsen\}@imada.sdu.dk}
}
\affil[2]{
	Institute for Theoretical Chemistry, University of Vienna, Wien A-1090, Austria\\
	\email{xtof@tbi.univie.ac.at}
}
\affil[3]{
	Bioinformatics Group, Department of Computer Science,and  Interdisciplinary Center for   Bioinformatics, University of Leipzig, Leipzig D-04107, Germany\\
	\email{stadler@bioinf.uni-leipzig.de}
}
\affil[4]{
	Max Planck Institute for Mathematics in the Sciences, Leipzig D-04103, Germany
}
\affil[5]{
	Fraunhofer Institute for Cell Therapy and Immunology, Leipzig D-04103, Germany
}
\affil[6]{
	Center for non-coding RNA in Technology and Health,  University of Copenhagen, Frederiksberg C DK-1870, Denmark
}
\affil[7]{
	Santa Fe Institute, 1399 Hyde Park Rd, Santa Fe NM 87501, USA
}
\affil[8]{
	Research Network Chemistry Meets Microbiology, University of Vienna,
	Wien A-1090, Austria
}

\date{}

\begin{document}
\maketitle

\begin{abstract}
  Graph transformation systems have the potential to be realistic models of
  chemistry, provided a comprehensive collection of reaction rules can be
  extracted from the body of chemical knowledge. A first key step for rule
  learning is the computation of atom-atom mappings, i.e., the atom-wise
  correspondence between products and educts of all published chemical
  reactions. This can be phrased as a maximum common edge subgraph problem
  with the constraint that transition states must have cyclic structure.
  We describe a search tree method well suited for small edit distance and
  an integer linear program best suited for general instances and
  demonstrate that it is feasible to compute atom-atom maps at large scales
  using a manually curated database of biochemical reactions as an
  example. In this context we address the network completion problem. 
\end{abstract}

\section{Introduction}
The individual records in databases of chemical reactions typically
describe, apart from more or less detailed meta-information, the
transformation of a set of educts into a set of products
\cite{Warr:14,Wittig:14}.  Both the product and the educt molecules have
representations as labeled graphs, where vertices designate atoms and edges
refer to chemical bonds. Chemical reactions therefore may be understood as
transformations of not necessarily connected graphs
\cite{Benkoe:03b,Yadav:04}. Chemical graph transformations must respect the
fundamental conservation principles of matter and charge and therefore
imply the existence of a bijection between vertex sets (atoms) of the
educts and products which is commonly known as the atom-atom map (AAM).

Chemical graph transformation are by no means arbitrary even when the
conservation laws imposed by the underlying physics are respected. Instead,
they conform to a large, but presumably finite, set of rules which in
chemistry are collectively known as reaction mechanism and ``named
reactions''. A chemical reaction partitions the sets of atoms and
  bonds of the participating molecules into a \emph{reaction center}
  comprising the bonds that change during the reactions and their incident
  atoms, and an remainder that is left unchanged. By virtue of being a
  bijection of the vertex (atom) sets, the AAM unambiguously determines the
  bonds that differ between educt and product molecules and thus it
  identifies the reaction center. The restriction of a chemical
  transformation to the reaction center, on the other hand, serves as
  minimal description of the underlying reaction rule.

The task to infer
transformation rules from empirical chemical knowledge therefore would be
greatly facilitated if each known reactions, i.e., each concrete pair of
educt and product molecules would imply a unique graph transformation.
Unfortunately, the true AAM is unknown in general, and even where the
chemical mechanism, and thus the actual graph transformation, has been
reported in the chemical literature, this information is in general not
stored together with the educt/product pair in a database. The inference of
chemical reaction mechanisms therefore requires that we first solve the
problem of inferring AAMs for the known chemical reactions.

Several computational methods for the AAM problem have been devised and
tested in the past \cite{Chen:13}. The most common formulations are
variants of the maximum common subgraph (isomorphism) problem
\cite{Ehrlich:2011}. In the NP-complete Maximum Common Edge Subgraph (MCES)
variant an isomorphic subgraphs of both the educt and product graph with a
maximal number of edges is identified. An alternative formulation as
Maximum Common Induced Subgraph (MCIS) problem \cite{Akutsu04} is also NP
complete. Algorithmic solutions decompose the molecules until only
isomorphic sub-graphs remain \cite{Akutsu04,Crabtree:2010}. 
In the context of graph transformation systems, few methods to infer transformation rules have been published \cite{Jeltsch:91}, and none applicable in the context of AAMs.

Neither solutions of MCES nor MCIS necessarily describe the true atom map,
however.  There is no reason why the re-organization of chemical bonds in a
chemical reaction should maximize a subgraph problem. Instead, they follow
strict rules that are codified, e.g., in the theory of imaginary
transition states (ITS) \cite{Fujita:86,Hendrickson:97}. There is only a
limited number of ITS ``layouts'' for single step reactions, corresponding
to the cyclic electron redistribution pattern usually involving less than
10 atoms \cite{Herges:1994}. In a most basic case, an elementary reaction,
the broken and newly formed bonds form an alternating cycle of a length
rarely exceeding $6$ or at most $8$ \cite{Hendrickson:97}. In
\cite{Mann:14a} we made use of this chemical constraint to devise a
Constraint Programming approach for elementary homovalent reactions, i.e.,
those chemical transformations that do not change the charge and oxidation
state of an atom. Here, we use an extended representation of chemical
graphs that explicitly represents lone pairs and bond orders; in this
manner the graph representation incorporates more detailed chemical
information.

Advances in bioinformatics technologies made it possible to infer
large-scale metabolic networks automatically from genomic information
\cite{Feist:2009,Biggs:15,Prigent:14}. These network models, however,
suffer from structural gaps in pathways \cite{Breitling:2008,Schaub:09},
caused by orphan metabolic activities, for which no sequences are known
and which cannot be inferred from genomic data. Thus there is an
  urgent need to infer missing metabolic reactions by other means.  We
  illustrate the potential of AAM for the discovery of novel metabolic
  reactions. To this end we determine whether chemically plausible AAMs can
  be founds connecting hypothetical educt/product pairs each consisting of
  one or two known metabolites.

\section{Chemical Reactions are Cyclic}
We model each molecule as a labeled, edge-weighted graph with loops.  While
the graph model used here is similar to most other formalizations of
chemical graphs, it differs in several subtle, but important, details, such as the way charges and lone pairs are modeled:

\begin{definition}[Molecule Graph]\label{defn:molgraph}
  A molecule graph $G=(V, E, l, w)$ is a labeled, edge-weighted, undirected
  graph with loops. The label function $l\colon V\cup E\rightarrow
  \Sigma_V \cup \Sigma_E$ denotes vertex and edge labels, and the weight
  function $w\colon E\rightarrow \mathbb{Z}$ denotes the weight of edges.
  These are assigned so that 
\begin{itemize}
\item Atoms are vertices, with labels denoting which type of atom.
\item Bonds are edges, with labels denoting the bond type and a weight
  encodes the number of involved electron pairs. Hence $1$, $2$, and $3$
  corresponds to single, double and triple bonds.
\item Lone pairs, i.e., pairs of non-bonding electrons, are modeled by 
  loops. Again the weight refers to the number of lone pairs.
\item Charges are modeled using a single special vertex together with edges
  from this special vertex to the charged atoms. The edge weight equals the
  atom's charge.
\item Free radicals, single non-bonding electrons, are modeled using a single special vertex together with edges from this special vertex to the atom with the free radical. The edge weight equals the number of free radicals.
\item Aromatic complexes are modeled by adding a special vertex for each aromatic complex in the molecules. Each atom participating in the aromatic complex has an edge to the special vertex with weight equal to the number of electrons at the atom taking part in the aromatic complex. The aromatic bonds themselves are edges with weight one, but are distinguished from single bonds by the edge label.
\end{itemize}
\end{definition}
See Fig.~\ref{fig:example} for example of molecule graph. See Fig.~\ref{fig:aromatic_ex} in Appendix \ref{app:aromatic} for an example of how the modeling of aromatic complexes works.
\begin{figure}[h!tb]
\begin{center}
\def\L{1}

\begin{tikzpicture}
    \node (c3) [dummy] at (0*\L, 1*\L)          {};
    \node (o3) at (0.866*\L, 0.5*\L)    {O$^-$};
    \node (o1) at (-0.866*\L, -0.5*\L)  {O};
    \node (c2) [dummy] at (-0.866*\L, 0.5*\L)   {};
    \node (o2) at (0*\L, 2*\L)          {O};
    \node (c1) at (-2*\L, 1*\L)         {H$_3$C};
    \path
        (c1) edge [bond1] (c2)
        (c2) edge [bond2] (o1)
        (c2) edge [bond1] (c3)
        (c3) edge [bond2] (o2)
        (c3) edge [bond1] (o3)
    ;
\end{tikzpicture}
\qquad
\qquad
\begin{tikzpicture}[every loop/.style={}]
    \node (c3) at (0*\L, 1*\L)          {C};
    \node (o3) at (0.866*\L, 0.5*\L)    {O};
    \node (o1) at (-0.866*\L, -0.5*\L)  {O};
    \node (c2) at (-0.866*\L, 0.5*\L)   {C};
    \node (o2) at (0*\L, 2*\L)          {O};
    \node (c1) at (-2*\L, 1*\L)         {C};
    \node (h1) at (-2*\L -1.134*\L, 1*\L + 0.5*\L)         {H};
    \node (h2) at (-2*\L -0.5*\L, 1*\L -1.134*\L)         {H};
    \node (h3) at (-2*\L + 0.5*\L, 1*\L + 1.134*\L)         {H};
    \node (charge) at (0.866*\L + \L, 0.5*\L) {$c$};
    \path
        (c1) edge [bond1] (c2)
        (c2) edge [bond2] (o1)
        (c2) edge [bond1] (c3)
        (c3) edge [bond2] (o2)
        (c3) edge [bond1] (o3)
        (c1) edge [bond1] (h1)
        (c1) edge [bond1] (h2)
        (c1) edge [bond1] (h3)
        (o3) edge [bond1, dashed] node [above] {$-1$} (charge)
        (o3) edge [bond1, in=-55, out=-125, loop, min distance=20] (o3)
    ;
\end{tikzpicture}

\end{center}
\caption{Usual depiction and molecule graph for pyruvate. Edge labels omitted. Edge weights shown by number of parallel edges (except where negative).}
\label{fig:example}
\end{figure}
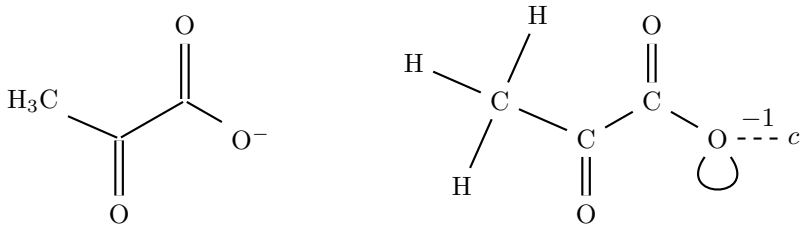

In the following two definitions
it will be convenient to consider instead of $E$ the set $E^*$ of all
possible edges on $V$ with edges in $e\in E^*\setminus E$ having weight
$w(e)=0$.

\begin{definition}[Atom-Atom Mapping]\label{defn:aam}
  Given two molecule graphs $G_1 = (V_1,$ $ E_1, l_1, w_1)$ and
  $G_2=(V_2, E_2, l_2, w_2)$, an atom-atom mapping from $G_1$ to $G_2$
  is a bijection $\psi\colon V_1 \rightarrow V_2$ that preserves
  vertex labels, i.e., $l_1(v) = l_2(\psi(v))$ for all $v\in V_1$.
  With $\psi$ we associate the \emph{cost}
  $c[\psi]= \sum_{e\in E_1^*} |w_2(\psi(e)) - w_1(e)|$.
\end{definition}
The cost measures the total number of electron pairs by which $G_1$
and $G_2$ differ w.r.t.\ to a given AAM.  Minimizing $c(\psi)$ can be
seen as an edit problem \cite{Justice:06,Riesen:09,Gao:10} and is
equivalent to the $NP$-hard MCES problem problem
\cite{Akutsu:13,Chen:13,Ehrlich:2011,Raymond:02,Bahiense:12}. Here we
are only interested in MCES instances that correspond to balanced
chemical reactions. The complexity results, however, also remains
valid also in this case. Next we investigate in some more detail what
exactly changes between $G_1$ and $G_2$ when an AAM $\psi$ is fixed.
\begin{definition}[Transition State]
  The transition state of an AAM $\psi: G_1 \to G_2$
  is the edge weighted graph $T_\psi=(V_\psi, E_\psi, w_\psi)$ 
  where $E_\psi = \{e\in E_1^*\mid w_1(e) \neq w_2(\psi(e))\}$, 
  $w_\psi(e) = w_2(\psi(e))- w_1(e)$, and $V_\psi\subseteq V_1$ 
  are all vertices incident to edges in $E_\psi$.
\end{definition}

By construction of molecule graphs, the weight of each edge is the number
of valence electrons. The atom type, i.e., the label of a vertex
determines the weighted degree $d_w(v) = \sum_{e\in \delta(v)} w(e)$. Here,
loops are counted twice. This reflects that the two electrons per bond order
are shared between the incident atoms, while both electrons of a lone pair
belong to the same atom. As a consequence, $d_w(v)$ is invariant under
all chemically acceptable atom maps. This restriction has important
consequences for the structure of transition states:

\begin{proposition}[Cyclic Transition States]\label{thm:cyclic}
  The transition state $T_\psi$ of an AAM $\psi$ can be decomposed into a
  collection of (not necessarily vertex disjoint) cycles $C_1, C_2, \dots,
  C_k$ with weights $w_{C_1}$, $w_{C_2}$, $\dots,w_{C_k}$ that are
  alternating between $+1$ and $-1$ along the cycles such that $w_\psi(e) =
  \sum_{i=1}^k w_{C_i}(e)$ for all $e\in E_\psi$.
\end{proposition}
\begin{proof}
  Since AAMs preserve vertex labels and vertex labels imply weighted degree
  the ``zero-flux condition'' $\sum_{e\in\delta(v)} w_\psi(e) = 0$ holds
  for all $v\in V_\psi$. We consider the following algorithm to construct a
  cycle $C$. Starting from a vertex $v$ we choose an $\{v,u\}$, with
  $w_\psi(\{v,u\}) > 0$, decrement $w_\psi(\{v,u\})$ by one and add
  $\{v,u\}$ to $C$. The vertex $u$ must be incident to an edge $\{u,w\}$
  with $w_\psi(\{u,w\}) < 0$, since otherwise the weighted valence would
  not be constant under $\psi$. We increase $w_\psi(\{u,w\})$ by one and add $\{u,w\}$
  to $C$. The process is repeated until we return to $v$, which is
  guaranteed by the finiteness of $V$. Clearly, $C$ is an Eulerian graph,
  i.e., all its vertex degrees are even. The procedure is repeated until no
  edges with $w_\psi \neq 0$ is left.  If $C$ contains a vertex with degree
  larger than two, we repeat the procedure recursively on $C$ until we are
  left with elementary cycles only.  
\end{proof}

The (weighted) degree $\delta_{\psi}(v):=\sum_{e:v\in e} |w_\psi(e)|$ of a
vertex in $T_\psi$ is even because in each step of the proof the value of
$\delta_{\psi}(v)$ is reduced by $2$.  Thus $T_\psi$ is a generalization of
an Eulerian graph, and Prop.~\ref{thm:cyclic} is the corresponding variant of
Veblen's theorem \cite{Veblen:1912}, which states that a graph is Eulerian
if and only if it is an edge-disjoint union of cycles.

\section{Finding Atom-Atom Mappings}

The cyclic nature of the transition states established in
Prop.~\ref{thm:cyclic} inspires two methods for finding minimum cost AAMs
described below. The idea was used in \cite{Mann:14a} in a much more restrictive
chemical setting.

\subsection{{\altcyc} --- A Search Tree Approach}
The idea of {\altcyc} is to construct a candidate transition state with a
given cost $\ell$ in a stepwise fashion and to simultaneously map $V_1$ to
$V_2$. The search for transition states proceeds depth first.  The validity
of a candidate is then checked by testing whether $G_1\setminus E_\psi$ and
$G_2\setminus \psi(E_\psi)$ are isomorphic. Finally, the parameter $\ell$
is increased until a valid mapping is found. A recursive definition of
{\altcyc} is given as Algorithm~\ref{alg:altcyc_disjoint}.

\begin{algorithm}[H]
\caption{\texttt{AltCyc}$(\psi, P, k, \sigma)$}
\label{alg:altcyc_disjoint}
\begin{algorithmic}
	\IF{$k = 1$}
		\IF{$w_1(P\!\text{.head},P\!\text{.tail}) + \sigma = w_2(\psi(P\!\text{.head}),\psi(P\!\text{.tail}))$}
			\STATE \texttt{Complete}$(\psi, P)$
		\ENDIF
	\ELSE
		\FOR{$i \in V_1 ~\land~ i\notin \dom(\psi)$}
			\FOR{$p \in V_2 ~\land~ p \notin \ran(\psi)$}
					\IF{$l_1(i)=l_2(p) ~\land~ w_1(P\!\text{.head},i) + \sigma =w_2(\psi(P\!\text{.head}),p) $}
						\STATE $\psi \leftarrow \psi \cup \{i\mapsto p\}$
						\STATE \texttt{AltCyc}$(\psi, P\!\text{.append}(i), k -1 , -1 \cdot \sigma)$
				\ENDIF
			\ENDFOR
		\ENDFOR
	\ENDIF
\end{algorithmic}
\end{algorithm}

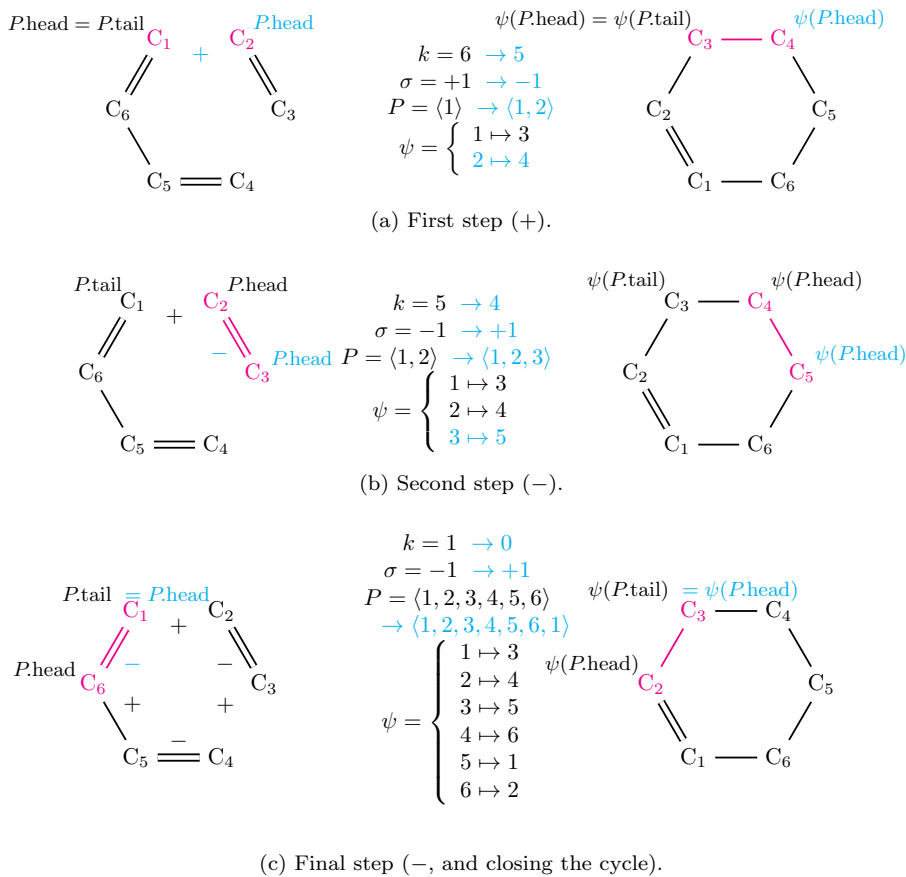
\begin{figure}[h!tb]
\begin{center}
\def\L{1.3}
\def\W{8.5}
\def\C{cyan}
\def\CC{magenta}

\subfloat[First step ($+$).]{\label{fig:example_a}
\resizebox{\textwidth}{!}{%
\begin{tikzpicture}
    \node (dummy1) at (-2*\L,0) {};
    \node (dummy2) at (\W+2*\L,0) {};

    \node (a3) at (\L,0) {C$_3$};
    \node [color=\CC, label={[label distance=-9pt, font=\small, color=\C]45:$P\!\text{.head}$}] (a2) at (\L*0.5,\L*0.866) {C$_2$};
    \node [label={[label distance=-7pt, font=\small]93:$P\!\text{.head}=P\!\text{.tail}~$}, color=\CC] (a1) at (\L*-0.5,\L*0.866) {C$_1$};
    \node (a6) at (\L*-1,0) {C$_6$};
    \node (a5) at (\L*-0.5,\L*-0.866) {C$_5$};
    \node (a4) at (\L*0.5,\L*-0.866) {C$_4$};
    \path
        (a1) edge [bond0] node [below, color=\C] {$+$} (a2)
        (a2) edge [bond2] (a3)
        (a3) edge [bond0] (a4)
        (a4) edge [bond2] (a5) 
        (a5) edge [bond1] (a6)
        (a6) edge [bond2] (a1);

    \node [align=center] (i1) at (\W*0.5,0) {
        $k=6 \color{\C}{ ~\rightarrow 5}$\\
        $\sigma=+1 \color{\C}{ ~\rightarrow -1}$\\
        $P = \langle 1 \rangle \color{\C}{ ~\rightarrow \langle 1, 2\rangle}$\\
        $\psi=\left\{\begin{array}{l}

            1 \mapsto 3\\

            \color{\C}{2 \mapsto 4}

        \end{array}\right.$};

    \node (b5) at (\L+\W,0) {C$_5$};
    \node [color=\CC, label={[label distance=-9pt, font=\small, color=\C]45:$\psi(P\!\text{.head})$}] (b4) at (\L*0.5+\W,\L*0.866) {C$_4$};
    \node [label={[label distance=-7pt, font=\small]93:$\psi(P\!\text{.head})=\psi(P\!\text{.tail})$}, color=\CC] (b3) at (\L*-0.5+\W,\L*0.866) {C$_3$};
    \node (b2) at (\L*-1+\W,0) {C$_2$};
    \node (b1) at (\L*-0.5+\W,\L*-0.866) {C$_1$};
    \node (b6) at (\L*0.5+\W,\L*-0.866) {C$_6$};
    \path
        (b1) edge [bond2] (b2)
        (b2) edge [bond1] (b3)
        (b3) edge [bond1, color=\CC] (b4)
        (b4) edge [bond1] (b5)
        (b5) edge [bond1] (b6)
        (b6) edge [bond1] (b1);
\end{tikzpicture}
}
}

\subfloat[Second step ($-$).]{\label{fig:example_b}
\resizebox{\textwidth}{!}{%
\begin{tikzpicture}
    \node (dummy1) at (-2*\L,0) {};
    \node (dummy2) at (\W+2*\L,0) {};

    \node [color=\CC, label={[label distance=-7pt, font=\small, color=\C]5:$P\!\text{.head}$}] (a3) at (\L,0) {C$_3$};
    \node [label={[label distance=-9pt, font=\small]45:$P\!\text{.head}$}, color=\CC] (a2) at (\L*0.5,\L*0.866) {C$_2$};
    \node [label={[label distance=-9pt, font=\small]135:$P\!\text{.tail}$}] (a1) at (\L*-0.5,\L*0.866) {C$_1$};
    \node (a6) at (\L*-1,0) {C$_6$};
    \node (a5) at (\L*-0.5,\L*-0.866) {C$_5$};
    \node (a4) at (\L*0.5,\L*-0.866) {C$_4$};
    \path
        (a1) edge [bond0] node [below] {$+$} (a2)
        (a2) edge [bond2, color=\CC] node [below left, color=\C] {$-$}  (a3)
        (a3) edge [bond0] (a4)
        (a4) edge [bond2] (a5) 
        (a5) edge [bond1] (a6)
        (a6) edge [bond2] (a1);

    \node [align=center] (i1) at (\W*0.5,0) {
        $k=5 \color{\C}{ ~\rightarrow 4}$\\
        $\sigma=-1 \color{\C}{ ~\rightarrow +1}$\\
        $P = \langle 1,2 \rangle \color{\C}{ ~\rightarrow \langle 1, 2,3\rangle}$\\
        $\psi=\left\{\begin{array}{l}
            1 \mapsto 3\\
            2 \mapsto 4\\

            \color{\C}{3 \mapsto 5}
        \end{array}\right.$};

    \node [color=\CC, label={[label distance=-7pt, font=\small, color=\C]5:$\psi(P\!\text{.head})$}] (b5) at (\L+\W,0) {C$_5$};
    \node [label={[label distance=-9pt, font=\small]45:$\psi(P\!\text{.head})$}, color=\CC] (b4) at (\L*0.5+\W,\L*0.866) {C$_4$};
    \node [label={[label distance=-9pt, font=\small]135:$\psi(P\!\text{.tail})$}] (b3) at (\L*-0.5+\W,\L*0.866) {C$_3$};
    \node (b2) at (\L*-1+\W,0) {C$_2$};
    \node (b1) at (\L*-0.5+\W,\L*-0.866) {C$_1$};
    \node (b6) at (\L*0.5+\W,\L*-0.866) {C$_6$};
    \path
        (b1) edge [bond2] (b2)
        (b2) edge [bond1] (b3)
        (b3) edge [bond1] (b4)
        (b4) edge [bond1, color=\CC] (b5)
        (b5) edge [bond1] (b6)
        (b6) edge [bond1] (b1);
\end{tikzpicture}
}
}

\subfloat[Final step ($-$, and closing the cycle).]{\label{fig:example_c}
\resizebox{\textwidth}{!}{%
\begin{tikzpicture}
    \node (dummy1) at (-2*\L,0) {};
    \node (dummy2) at (\W+2*\L,0) {};

    \node (a3) at (\L,0) {C$_3$};
    \node (a2) at (\L*0.5,\L*0.866) {C$_2$};
    \node [label={[label distance=-7pt, font=\small]90:$P\!\text{.tail}\color{\C}{ ~= P\!\text{.head}}$}, color=\CC] (a1) at (\L*-0.5,\L*0.866) {C$_1$};
    \node [label={[label distance=-8pt, font=\small]170:$P\!\text{.head}\;$}, color=\CC] (a6) at (\L*-1,0) {C$_6$};
    \node (a5) at (\L*-0.5,\L*-0.866) {C$_5$};
    \node (a4) at (\L*0.5,\L*-0.866) {C$_4$};
    \path
        (a1) edge [bond0] node [below] {$+$} (a2)
        (a2) edge [bond2] node [below left] {$-$} (a3)
        (a3) edge [bond0] node [above left] {$+$} (a4)
        (a4) edge [bond2] node [above] {$-$} (a5) 
        (a5) edge [bond1] node [above right] {$+$} (a6)
        (a6) edge [bond2, color=\CC] node [below right, color=\C] {$-$} (a1);

    \node [align=center] (i1) at (\W*0.5,0) {
        $k=1 \color{\C}{ ~\rightarrow 0}$\\
        $\sigma=-1 \color{\C}{ ~\rightarrow +1}$\\
        $P = \langle 1,2,3,4,5,6 \rangle$ \\
        $\color{\C}{ ~~~~\rightarrow \langle 1, 2,3,4,5,6,1\rangle}$\\
        $\psi=\left\{\begin{array}{l}

            1 \mapsto 3\\
            2 \mapsto 4\\
            3 \mapsto 5\\

            4 \mapsto 6\\
            5 \mapsto 1\\

            6 \mapsto 2

        \end{array}\right.$\\
        ~\\~};

    \node (b5) at (\L+\W,0) {C$_5$};
    \node (b4) at (\L*0.5+\W,\L*0.866) {C$_4$};
    \node [label={[label distance=-7pt, font=\small]90:$\psi(P\!\text{.tail})\color{\C}{ ~= \psi(P\!\text{.head})}$}, color=\CC] (b3) at (\L*-0.5+\W,\L*0.866) {C$_3$};
    \node [label={[label distance=-8pt, font=\small]170:$\psi(P\!\text{.head})$}, color=\CC] (b2) at (\L*-1+\W,0) {C$_2$};
    \node (b1) at (\L*-0.5+\W,\L*-0.866) {C$_1$};
    \node (b6) at (\L*0.5+\W,\L*-0.866) {C$_6$};
    \path
        (b1) edge [bond2] (b2)
        (b2) edge [bond1, color=\CC] (b3)
        (b3) edge [bond1] (b4)
        (b4) edge [bond1] (b5)
        (b5) edge [bond1] (b6)
        (b6) edge [bond1] (b1);
\end{tikzpicture}
}
}

\end{center}
\caption{Stepwise execution of \altcyc. Cyan marks the changes within the step. Magenta marks the considered edges and incident vertices.
}
\label{fig:altcyc_ex}
\end{figure}

To explain the algorithm, we first restrict ourself to mappings with
transition states consisting of a single elementary cycle. The four
parameters are a partial atom-atom mapping candidate $\psi$, the partial
transition state $P$ constructed so far encoded as a list of vertices from
$V_1$, and the number $k$ of edges still to be identified, and the variable 
$\sigma\in \{-1, 1\}$ that determines whether the current step will add 
or remove weight.

The search in {\altcyc} starts from all pairs $(i,p)$ with $i\in V_1$ and
$p\in V_2$ with $l_2(p)=l_1(i)$; the map $\psi$ is initalized $\psi(i)=p$
and the path starts with $P=\{i\}$. W.l.o.g., the first step is a positive
change of weight, i.e., $\sigma=1$.  In each step in the algorithm, a new
pair $(i,p)\in V_1\times V_2$ with matching labels is found and if the
$w_1(\{P\!\text{.head}, i\})$ and $w_2(\{\psi(P\!\text{.head}), p\})$
differ by exactly one, $i$ is appended to $P$, $\psi$ is extended such
that $\psi(i) = p$ and the algorithm is called again with $k$ replaced by
$k-1$. If $k=1$ has been reached, it only remains to close the alternating
cycle. If this is possible, the candidate transition state is extended to a
full AAM where no further changes are allowed. To this end, a graph
isomorphism algorithm is used.  We use {\vftwo} \cite{Cordella:04} in
procedure \texttt{Complete} (see Appendix \ref{app:code}) because it
has the added benefit of using data structures that are similar to those
used in other parts of {\altcyc}.  The first two and the last step of
{\altcyc} applied to a Diels-Alder reaction are shown in Fig.~\ref{fig:altcyc_ex}.

In order to handle transition states that are connected but not elementary
cycles, as the case of a bi-cyclic or coarctate reaction
\cite{Hendrickson:97}, we modify {\altcyc} to allow weight differences
larger than one. Such vertices must then be revisited. In addition, we
disallow using the same edge with different signs of $\sigma$ because a
pair of such steps would cancel. The modified approach is outlined in
Algorithm~\ref{alg:altcyc_nondisjoint}. The key point is that we now need
to keep track of the weight changes, $w_P(e)$, that we have already made
along an edge $e$ (found using the procedure \texttt{WeightAlongPath}, see Appendix \ref{app:code}).
The condition for acceptable weight differences becomes
$w_1(e) + w_P(e) + \sigma \leq w_2(\psi(e))$ if a bond is added
($\sigma=1$), and $w_1(e) + w_P(e) + \sigma \geq w_2(\psi(e))$ for bond
subtraction ($\sigma=-1$).

\begin{algorithm}[H]
\caption{\texttt{AltCyc}$^*(\psi, P, k, \sigma)$}
\label{alg:altcyc_nondisjoint}
\begin{algorithmic}
  \STATE \textcolor{gray}{// As \texttt{AltCyc}$\dots$}
  \FOR{$(i,p) \in V_1\times V_2 \textbf{ with } l_1(i) = l_2(p)$}
     \IF{$i\notin \dom(\psi) ~\land~ p\notin \ran(\psi)$}
        \STATE \textcolor{gray}{// As \texttt{AltCyc}, 
               but using $\leq$ and $\geq\dots$}
     \ELSIF{$\psi(i) = p $}
        \STATE $w_P \leftarrow 
               \texttt{WeightAlongPath}(\{P\!\text{.head}, i\},P)$
     \IF{$w_P \geq 0 ~\land~ \sigma=1$}
        \IF{$w_1(P\!\text{.head}, i) + w_P + 
               \sigma \leq w_2(\psi(P\!\text{.head}), p)$}
           \STATE \texttt{AltCyc}$^*(\psi, P\!\text{.append}(i), k-1, -1 \cdot \sigma)$
        \ENDIF
        \ELSIF{$w_P \leq 0 ~\land~ \sigma=-1$}
           \STATE \textcolor{gray}{// Symmetric case$\dots$}
        \ENDIF
     \ENDIF
  \ENDFOR
\end{algorithmic}
\end{algorithm}

There is no guarantee that the transition state is connected. To accommodate
disconnected transition states it suffices to replace the path $P$ by a
list of paths, where the last path is the current path and all previous
paths are kept in order to correctly calculate $w_P(e)$. If a path closes
before $k=0$ is reached, the current cycle is completed and the algorithm
restarts to build new path from another initial vertex.

The stepwise approach in {\altcyc} naturally allows for an elucidation
of the mechanism underlying an AAM found by the algorithm. In Fig.~\ref{fig:reaction_ex}
 the automatic inference of such a mechanism is illustrated. Each step in the figure, the usual way of drawing arrow pushing diagrams, corresponds to two steps in {\altcyc}.

\begin{figure}[h!tb]
\begin{center}
\def\L{0.65}
\subfloat[Initial molecule.]{
\begin{tikzpicture}
    \node (c9) at (\L,0) {C};
    \node (c10) at (\L*0.5,\L*0.866) {C};
    \node (c5) at (\L*-0.5,\L*0.866) {C};
    \node (c6) at (\L*-1,0) {C};
    \node (c7) at (\L*-0.5,\L*-0.866) {C};
    \node (c8) at (\L*0.5,\L*-0.866) {C};
    \node (c4) at (\L*2*-0.5, \L*2*0.866) {C};
    \node (c1) at (\L*2*-1, \L*2*0) {C};
    \node (c3) at (\L*2*-0.5 - \L, \L*2*0.866) {C};
    \node (c2) at (\L*-0.5 - \L*2,\L*0.866) {C};
    \node (c11) at (\L*2*0.5,\L*2*0.866) {C};
    \node (o15) at (\L*2,0) {O};
    \node (c13) at (\L*0.5 + \L*2,\L*0.866) {C};
    \node (c12) at (\L*2*0.5 + \L,\L*2*0.866) {C};
    \node (o14) at (\L*0.5 + \L*2 + \L,\L*0.866) {O};
    \node (c0) at (\L*2*-1, \L*2*0 - \L) {C};
    \node (c18) at (\L*2*-1 - \L, \L*2*0 - \L*0.5) {C};
    \node (h37) at (\L*0.5 + \L*2,\L*-0.866) {H};
    \node (h37_var) at (\L*-0.5 - \L*2 -\L, \L*0.866) {};
    \path
        (c18) edge [bond1] (c1)
        (c0) edge [bond1] (c1)
        (c1) edge [bond2] (c2)
        (c2) edge [bond1] (c3)
        (c3) edge [bond1] (c4)
        (c4) edge [bond1] (c5)
        (c5) edge [bond2] (c6)
        (c6) edge [bond1] (c7)
        (c7) edge [bond1] (c8)
        (c8) edge [bond1] (c9)
        (c9) edge [bond2] (c10)
        (c10) edge [bond1] (c11)
        (c11) edge [bond1] (c12)
        (c12) edge [bond1] (c13)
        (c13) edge [bond1] (o15)
        (c13) edge [bond2] (o14)
        (c1) edge [bond0] (c6)
        (c5) edge [bond0] (c10)
        (c9) edge [bond0] (o15)
        (o15) edge [bond1] node (15-H) {} (h37)
        ;
    \path
        (15-H) edge [bend left, ->,  >=stealth, color=red] (c9) 
        ;
\end{tikzpicture}
}
\qquad
\subfloat[After two bond changes.]{
\begin{tikzpicture}
    \node (c9) at (\L,0) {C};
    \node (c10) at (\L*0.5,\L*0.866) {C};
    \node (c5) at (\L*-0.5,\L*0.866) {C};
    \node (c6) at (\L*-1,0) {C};
    \node (c7) at (\L*-0.5,\L*-0.866) {C};
    \node (c8) at (\L*0.5,\L*-0.866) {C};
    \node (c4) at (\L*2*-0.5, \L*2*0.866) {C};
    \node (c1) at (\L*2*-1, \L*2*0) {C};
    \node (c3) at (\L*2*-0.5 - \L, \L*2*0.866) {C};
    \node (c2) at (\L*-0.5 - \L*2,\L*0.866) {C};
    \node (c11) at (\L*2*0.5,\L*2*0.866) {C};
    \node (o15) at (\L*2,0) {O};
    \node (c13) at (\L*0.5 + \L*2,\L*0.866) {C};
    \node (c12) at (\L*2*0.5 + \L,\L*2*0.866) {C};
    \node (o14) at (\L*0.5 + \L*2 + \L,\L*0.866) {O};
    \node (c0) at (\L*2*-1, \L*2*0 - \L) {C};
    \node (c18) at (\L*2*-1 - \L, \L*2*0- \L*0.5) {C};
    \node (h37) at (\L*0.5 + \L*2,\L*-0.866) {H$^+$};
    \node (h37_var) at (\L*-0.5 - \L*2 -\L, \L*0.866) {};
    \path
        (c18) edge [bond1] (c1)
        (c0) edge [bond1] (c1)
        (c1) edge [bond2] (c2)
        (c2) edge [bond1] (c3)
        (c3) edge [bond1] (c4)
        (c4) edge [bond1] (c5)
        (c5) edge [bond2] (c6)
        (c6) edge [bond1] (c7)
        (c7) edge [bond1] (c8)
        (c8) edge [bond1] (c9)
        (c9) edge [bond2] node (9-10) {} (c10)
        (c10) edge [bond1] (c11)
        (c11) edge [bond1] (c12)
        (c12) edge [bond1] (c13)
        (c13) edge [bond1] (o15)
        (c13) edge [bond2] (o14)
        (c1) edge [bond0] (c6)
        (c5) edge [bond0] (c10)
        (c9) edge [bond1] (o15)
        (o15) edge [bond0] node (15-H) {} (h37)
        ;
    \path
        (9-10) edge [bend left, ->,  >=stealth, color=red] (c5) 
        ;
\end{tikzpicture}
}

\subfloat[After four bond changes.]{
\begin{tikzpicture}
    \node (c9) at (\L,0) {C};
    \node (c10) at (\L*0.5,\L*0.866) {C};
    \node (c5) at (\L*-0.5,\L*0.866) {C};
    \node (c6) at (\L*-1,0) {C};
    \node (c7) at (\L*-0.5,\L*-0.866) {C};
    \node (c8) at (\L*0.5,\L*-0.866) {C};
    \node (c4) at (\L*2*-0.5, \L*2*0.866) {C};
    \node (c1) at (\L*2*-1, \L*2*0) {C};
    \node (c3) at (\L*2*-0.5 - \L, \L*2*0.866) {C};
    \node (c2) at (\L*-0.5 - \L*2,\L*0.866) {C};
    \node (c11) at (\L*2*0.5,\L*2*0.866) {C};
    \node (o15) at (\L*2,0) {O};
    \node (c13) at (\L*0.5 + \L*2,\L*0.866) {C};
    \node (c12) at (\L*2*0.5 + \L,\L*2*0.866) {C};
    \node (o14) at (\L*0.5 + \L*2 + \L,\L*0.866) {O};
    \node (c0) at (\L*2*-1, \L*2*0 - \L) {C};
    \node (c18) at (\L*2*-1 - \L, \L*2*0- \L*0.5) {C};
    \node (h37) at (\L*0.5 + \L*2,\L*-0.866) {};
    \node (h37_var) at (\L*-0.5 - \L*2 -\L, \L*0.866) {H$^+$};
    \path
        (c18) edge [bond1] (c1)
        (c0) edge [bond1] (c1)
        (c1) edge [bond2] (c2)
        (c2) edge [bond1] (c3)
        (c3) edge [bond1] (c4)
        (c4) edge [bond1] (c5)
        (c5) edge [bond2] node (5-6) {} (c6)
        (c6) edge [bond1] (c7)
        (c7) edge [bond1] (c8)
        (c8) edge [bond1] (c9)
        (c9) edge [bond1] node (9-10) {} (c10)
        (c10) edge [bond1] (c11)
        (c11) edge [bond1] (c12)
        (c12) edge [bond1] (c13)
        (c13) edge [bond1] (o15)
        (c13) edge [bond2] (o14)
        (c1) edge [bond0] (c6)
        (c5) edge [bond1] (c10)
        (c9) edge [bond1] (o15)
        (o15) edge [bond0] node (15-H) {} (h37)
        ;
    \path
        (5-6) edge [bend right, ->,  >=stealth, color=red] (c1) 
        ;
\end{tikzpicture}
}
\qquad
\subfloat[After six bond changes.]{
\begin{tikzpicture}
    \node (c9) at (\L,0) {C};
    \node (c10) at (\L*0.5,\L*0.866) {C};
    \node (c5) at (\L*-0.5,\L*0.866) {C};
    \node (c6) at (\L*-1,0) {C};
    \node (c7) at (\L*-0.5,\L*-0.866) {C};
    \node (c8) at (\L*0.5,\L*-0.866) {C};
    \node (c4) at (\L*2*-0.5, \L*2*0.866) {C};
    \node (c1) at (\L*2*-1, \L*2*0) {C};
    \node (c3) at (\L*2*-0.5 - \L, \L*2*0.866) {C};
    \node (c2) at (\L*-0.5 - \L*2,\L*0.866) {C};
    \node (c11) at (\L*2*0.5,\L*2*0.866) {C};
    \node (o15) at (\L*2,0) {O};
    \node (c13) at (\L*0.5 + \L*2,\L*0.866) {C};
    \node (c12) at (\L*2*0.5 + \L,\L*2*0.866) {C};
    \node (o14) at (\L*0.5 + \L*2 + \L,\L*0.866) {O};
    \node (c0) at (\L*2*-1, \L*2*0 - \L) {C};
    \node (c18) at (\L*2*-1 - \L, \L*2*0- \L*0.5) {C};
    \node (h37) at (\L*0.5 + \L*2,\L*-0.866) {};
    \node (h37_var) at (\L*-0.5 - \L*2 -\L, \L*0.866) {H$^+$};
    \path
        (c18) edge [bond1] (c1)
        (c0) edge [bond1] (c1)
        (c1) edge [bond2] node (1-2) {} (c2)
        (c2) edge [bond1] (c3)
        (c3) edge [bond1] (c4)
        (c4) edge [bond1] (c5)
        (c5) edge [bond1] node (5-6) {} (c6)
        (c6) edge [bond1] (c7)
        (c7) edge [bond1] (c8)
        (c8) edge [bond1] (c9)
        (c9) edge [bond1] node (9-10) {} (c10)
        (c10) edge [bond1] (c11)
        (c11) edge [bond1] (c12)
        (c12) edge [bond1] (c13)
        (c13) edge [bond1] (o15)
        (c13) edge [bond2] (o14)
        (c1) edge [bond1] (c6)
        (c5) edge [bond1] (c10)
        (c9) edge [bond1] (o15)
        (o15) edge [bond0] node (15-H) {} (h37)
        ;
    \path
        (1-2) edge [bend left, ->,  >=stealth, color=red] (h37_var) 
        ;
\end{tikzpicture}
}

\subfloat[Resulting molecule.]{
\begin{tikzpicture}
    \node (c9) at (\L,0) {C};
    \node (c10) at (\L*0.5,\L*0.866) {C};
    \node (c5) at (\L*-0.5,\L*0.866) {C};
    \node (c6) at (\L*-1,0) {C};
    \node (c7) at (\L*-0.5,\L*-0.866) {C};
    \node (c8) at (\L*0.5,\L*-0.866) {C};
    \node (c4) at (\L*2*-0.5, \L*2*0.866) {C};
    \node (c1) at (\L*2*-1, \L*2*0) {C};
    \node (c3) at (\L*2*-0.5 - \L, \L*2*0.866) {C};
    \node (c2) at (\L*-0.5 - \L*2,\L*0.866) {C};
    \node (c11) at (\L*2*0.5,\L*2*0.866) {C};
    \node (o15) at (\L*2,0) {O};
    \node (c13) at (\L*0.5 + \L*2,\L*0.866) {C};
    \node (c12) at (\L*2*0.5 + \L,\L*2*0.866) {C};
    \node (o14) at (\L*0.5 + \L*2 + \L,\L*0.866) {O};
    \node (c0) at (\L*2*-1, \L*2*0 - \L) {C};
    \node (c18) at (\L*2*-1 - \L, \L*2*0- \L*0.5) {C};
    \node (h37) at (\L*0.5 + \L*2,\L*-0.866) {};
    \node (h37_var) at (\L*-0.5 - \L*2 -\L, \L*0.866) {H};
    \path
        (c18) edge [bond1] (c1)
        (c0) edge [bond1] (c1)
        (c1) edge [bond1] node (1-2) {} (c2)
        (c2) edge [bond1] (c3)
        (c3) edge [bond1] (c4)
        (c4) edge [bond1] (c5)
        (c5) edge [bond1] node (5-6) {} (c6)
        (c6) edge [bond1] (c7)
        (c7) edge [bond1] (c8)
        (c8) edge [bond1] (c9)
        (c9) edge [bond1] node (9-10) {} (c10)
        (c10) edge [bond1] (c11)
        (c11) edge [bond1] (c12)
        (c12) edge [bond1] (c13)
        (c13) edge [bond1] (o15)
        (c13) edge [bond2] (o14)
        (c1) edge [bond1] (c6)
        (c5) edge [bond1] (c10)
        (c9) edge [bond1] (o15)
        (o15) edge [bond0] node (15-H) {} (h37)
        (c2) edge [bond1] (h37_var)
        ;
\end{tikzpicture}
}

\end{center}
\caption{An example AAM for Stork's cyclisation of farnesyl acetic acid to
  ambreinolide \cite{Yoder:2005}. Note that only the single hydrogen in the
  transition state is shown, and while it is assumed in the model that it is the same
  hydrogen leaving and later entering, in actual chemistry it is a
  different hydrogen.}
\label{fig:reaction_ex}
\end{figure}
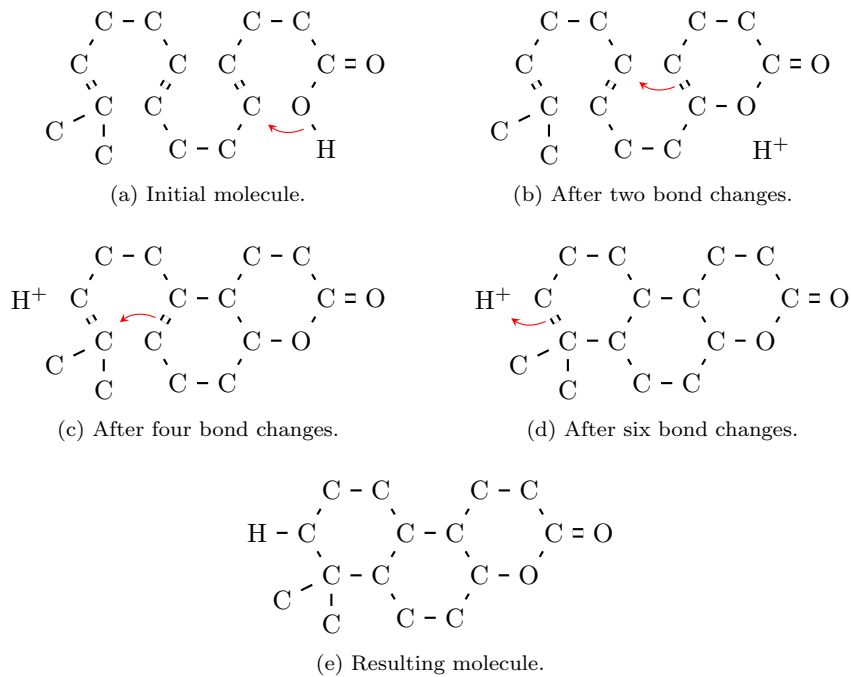

Taken together, {\altcyc} uses $O((n^2)^k) = O(n^{2k})$ recursive calls,
where $n=|V_1|=|V_2|$. Exploiting the fact that only edges to the special
vertex for a charge can be negative, this reduces to $O(n^{k + l})$, where
$l$ is the number of components in the transition state candidate, because
it suffices to examine only the $O(1)$ edges incident to $P\!\text{.head}$
or $\psi(P\!\text{.head})$ depending on whether we are making a negative or
positive step in the algorithm. In addition, {\altcyc} incurs the cost of 
the graph isomorphism check for completing the mapping.

In practice, however, the runtime is much lower since vertex labels
must match. The runtime nevertheless still depends heavily on $k$, and
thus the length of the optimal mapping of the instance. However, as
discussed $k$ can be assumed to be small for the case of inferring
chemical transformation rules. Due to depth first
strategy, the memory consumption of {\altcyc} is $O(n)$.

\subsection{{\ilptwo} --- An Integer Linear Program}
The AAM problem can also be phrased as an ILP. We use binary variables
$m_{ip}$ to encode the mapping $\psi$ as $m_{ip}=1$ iff $\psi(i) = p$ and
$m_{ip}=0$ for all other combinations of $i$ and $p$.  To enforce that $\psi$ is
vertex label preserving we set $m_{ip} = 0$ for $l_1(i) \neq l_2(p)$, and
to ensure $\psi$ is a bijection we formulate the following linear
constraints.
\begin{equation*}
\forall i\in V_1\colon \sum_{p\in V_2} m_{ip} = 1 
\qquad\textrm{and}\qquad 
\forall p\in V_2\colon \sum_{i\in V_1} m_{ip} = 1 
\end{equation*}
The most obvious way to proceed would be to keep track of the mapping
between the edge sets using either binary variables describing whether a
bond is mapped or not as in \cite{First:12}, or integer variables that
denote the weight change if a bond is mapped, and zero for unmapped bonds.
For such variables we would need $O(|V|^4)$ constraints,
however. Empirically we found that ILP-solvers quickly run out of memory and
become very slow for such a model.

Though there already exist ILP formulations of similar problems with only $O(|V|^2)$ constraints \cite{Latendresse:12}, obtained by exploiting the sparseness of molecule graohs, we propose a new ILP formulation based on the Kaufmann and Broeckx
linearization of the quadratic assignment problem \cite{Burkard:99}, which also
needs only $O(|V|^2)$ constraints.

We introduce integer variables $c^+_{ip}\in\mathbb{N}_0$ and
$c^-_{ip}\in\mathbb{N}_0$ that model the positive and negative weight
changes respectively of all edges incident to vertex $i\in V_1$ if $\psi(i)
= p$. Both $c^+_{ip}$ and $c^-_{ip}$ are zero for all other combinations of
$i$ and $p$. Making use of the fact that weight changes are balanced, i.e.\
$\sum_{e\in \delta(v)} w_\psi(e) = 0$ for all $v\in V_\psi$, we can use the
following constraint for all $i\in V_1$:
\begin{equation*}
  \sum_{p\in V_2}c^+_{ip} = \sum_{p\in V_2}c^-_{ip}
\end{equation*}
We also substitute them in the objective function:
\begin{equation*}
  \textit{obj} = \sum_{\mathclap{(i,p)\in V_1\times V_2}}\; 
  c^+_{ip} \;\;+\;\; \sum_{\mathclap{(i,p)\in V_1\times V_2}}\; c^-_{ip}
\end{equation*}
Since the change variables are included in the objective function they 
will implicitly be constrained from above. In order to constrain them 
from below we use the following constraints for all $(i,p)\in V_1\times V_2$:
\[c^+_{ip} \geq (m_{ip}-1)\cdot M + \sum_{\mathclap{(j,q)\in V_1\times V_2}}
  \;  m_{jq}\cdot\max\{0, w_2(\{p,q\}) - w_1(\{i,j\})\} \] 
\[c^-_{ip} \geq (m_{ip}-1)\cdot M + \sum_{\mathclap{(j,q)\in V_1\times V_2}}
  \;  m_{jq}\cdot\max\{0, w_1(\{i,j\}) - w_2(\{p,q\})\} \] 
where $M$ is a suitably large constant. It suffices to set $M$ to the 
largest weighted degree to void the constraint when $m_{ip}=0$.
The first term voids the constraints if $m_{ip}\neq  1$.
The sums correspond to the sum of all positive (negative) changes of edges incident to $i$ and $p$ respectively, if indeed these edges are mapped to each other. 

Unlike {\altcyc} we have little control over intermediate steps in the 
reaction, but using {\ilptwo} we have much freedom to modify the cost model 
used. Assuming we have an integer linear programming solver available 
{\ilptwo} takes very little time to implement.

\subsection{Enumeration of All Optimal Atom-Atom Mappings}
So far we have focused on the problem of finding a single AAM. The solution
of the optimization problem is in general not unique, however. A particular
problem in this context are symmetries of the educt or product molecules,
because this may bloat the number of AAMs. We are therefore interested only
in nonequivalent AAMs.

\begin{definition}[Equivalent Atom-Atom Mappings]\label{defn:equiv}
  For a given AAM define $G_\psi = (V_\psi, E_\psi, l_\psi)$ with
  vertex set $V_\psi = V_1$, edge set
  $E_\psi = E_1 \cup \psi^{-1}(E_2)$, and label function
  $l_\psi(x) = (l_1(x), l_2(\psi(x)))$. If $x\notin \dom(l_i)$ then
  $l_i(x) = \varepsilon_i$, where $\varepsilon_i$ is some label not in
  $\ran(l_i)$, denoting a non-edge. We say two atom-atom mappings,
  $\psi$ and $\varphi$ are equivalent if the graphs $G_\psi$ and
  $G_\varphi$ are isomorphic.
\end{definition}

Now, let us consider whether a transition state candidate of an atom-atom
mapping uniquely defines the full mapping.

\begin{definition}[Completion of Partial Mapping]
Given a partial AAM $\psi'\colon A\subset V_1 \rightarrow B\subset
V_2$, a completion of $\psi'$ is an AAM such that $\psi|_{A} = \psi'$ 
and outside $A$, $\psi$ preserves all properties of $G_1$ and $G_2$.
\end{definition}
Note that such a completion need not exist for a given partial AAM. 
\begin{proposition}[Partial Mapping]\label{thm:partial}
  If $\psi$ and $\varphi$ are two completions of a partial AAM $\psi'$,
  $\psi$ and $\varphi$ are equivalent.
\end{proposition}
\begin{proof}
  Consider the two AAMs $\psi$ and $\varphi$ and their associated graphs
  $G_\psi$ and $G_\varphi$. By assumption, they are both completions of the
  same partial AAM $\psi'$; therefore the two induced sub-graphs
  $G_\psi[\dom(\psi')]$ and $G_\varphi[\dom(\psi')]$ are identical.
  Consider the subgraphs $G':= G_\psi \setminus E(G_\psi[\dom(\psi')])$ and
  $G'':= G_\varphi \setminus E(G_\varphi[\dom(\psi')])$ without edges in
  $G_\psi[\dom(\psi')]$. 
  $G'$ and $G''$ both are identical to $G_1\setminus
  E(G_\psi[\dom(\psi')])$ if only considering the labels from $l_1$ 
  in each of $G_\psi$ and $G_\varphi$. As both $\psi$ and $\varphi$
  preserve all properties of $G_1$ and $G_2$ outside $\dom(\psi')$, the labels 
  from $l_2$ are always identical to the labels from $l_1$ outside
  $\dom(\psi')$.

  Thus $G_\psi$ and $G_\varphi$ are isomorphic and by definition $\psi$ and
  $\varphi$ are equivalent.  
\end{proof}

Prop.~\ref{thm:partial} can be applied in different ways.
In {\altcyc} it shows we only have to complete each candidate
transition state once in order to enumerate all mappings. In {\ilptwo} it
can be used to exclude solutions based on mapping variables defining the
transition states instead of all mapping variables.

\section{Results}
The RHEA \cite{Morgat:15} database (v.\ 50), which provides access to a
large set of expert-curated biochemical reactions, has been used to test
our suggested AAM algorithms, and to underline the necessity of graph
transformation methods for network completion. We exclude all reactions with
unspecified repeating units and wildcards, resulting in a set of $19753$
reactions involving a set, $M$, of $3786$ non-isomorphic molecular graphs. We performed
a statistical analysis of RHEA, that shows how often molecules are used in
the reaction listed in the database, and how many non-isomorphic isomers
are stored in RHEA. Interestingly, terpene chemistry \cite{Degenhardt:2009}
clearly dominates the high frequency isomers (see Appendix \ref{app:rhea}). Due
to space limitations, we focus on a brief runtime analysis and network
completion results. As {\altcyc} constructs solutions in a stepwise
fashion, a chemical mechanism explaining the bond changes as subsequent
transformations is naturally inferred. An example for a mechanistic inference of
Stork's cyclisation of farnesyl acetic acid to ambreinolide
\cite{Yoder:2005} is given in Fig.~\ref{fig:reaction_ex}.

\par\noindent\textbf{Runtime.} We compared {\altcyc}, {\ilptwo}, and 
a na{\"i}ve ILP-implementation with $O(n^4)$ constraints, \texttt{ILP4},
with regard to their ability of enumerating all non-equivalent AAMs within
a fixed runtime (see appendix \ref{app:runtime}). We found that {\ilptwo} drastically outperforms the
na{\"i}ve ILP-implementation and also is systematically more efficient than
{\altcyc}. The latter has a (small) advantage for instances with 
small transitions states.
For both methods we see an exponential decline in ratio of quickly solved instances as size of instances grow, this corresponds well with the expected exponential runtime.

\par\noindent\textbf{Network Completion.}
Databases of metabolic networks are by no means complete because the
enzymes catalyzing many of the reactions in particular in the so-called
secondary metabolism have remained unknown. Furthermore, for almost one
third of the known metabolic activities, no protein sequences are
  known that could encode the corresponding enzyme.  \emph{Network
  completion} is an important task to fill gaps i.e.\ missing reaction
steps, in genome-scale metabolic networks. 
Reaction perception, i.e.\ finding AAMs, is the only
technique capable of finding possible candidates for the missing reactions,
where homology based methods fail, due to lack of data.  

Inferring all
candidate \emph{2-to-2} reactions addresses this issue by determining for
all disjoint pairs $A,B$ of multisets $A$ (one or two educt molecules,
potentially isomorphic) and multisets $B$ (one or two product molecules, also potentially isomorphic),
whether there is a chemically plausible reaction transforming $A$ to $B$. By
Prop.~\ref{thm:cyclic}, {\em any} reaction satisfying mass and charge
balance has a cyclic transition state. 

Let $R_{2,2}$ denote the set of all sets $\{A,B\}$ such that $A$ and $B$ are disjoint multi-subsets of the set of molecules $M$, both of size at most $2$ with $A$ and $B$ containing the same vertex labels, charges, etc.
The set of test instances $R_{2,2}$
of \emph{2-to-2} reactions that satisfy mass balance can be extracted from a database with
molecule set $M$ in time $O(|M|^2 \log |M| + |R_{2,2}|$) using
Algorithm~\ref{alg:2to2} (see Appendix \ref{app:2to2}). We obtain a set of
$|R_{2,2}|=114,429,849$ balanced reaction candidates with at most two
molecules on either side of the reaction.

It is not feasible to test 100 million candidates for chemical feasibility
in an exact manner. Using the length of the transition state as a filter,
however, will remove implausible candidates as well as multi-step
mechanisms. The length of the transition state can be bounded in both
{\altcyc} as well as {\ilptwo}. We used {\altcyc} because of its performance
advantage with small transition states. In a random sample of 10.000
instances drawn from $R_{2,2}$ we found 34, 59, and 167 reactions with
transition states of length 4, 6, and 8, respectively.  Extrapolating from
this sample we have to expect approximately three million candidate
reactions with AAMs that will need to be examined in more detail. Clearly
this number is too large for a biochemical network. Further
pruning of the candidate list will thus require additional information,
e.g., on the energetics of the reactions and on these reaction mechanisms
plausibly catalyzed by enzymes. However, it underlines the need for
graph transformation techniques for computing realistic candidate sets.   

\section{Conclusion}
Graph transformation systems have great potential as a model of chemistry in
  particular in the context of large reaction networks. Their practical
  usefulness, however, stands and falls with the ability to produce
  collections of transformation rules that closely reflect chemical
  reality. We have shown here that the extraction of AAMs from
  educt/product pairs is a necessary first step because the restriction of
  the graph transformation to the reaction center, which is uniquely
  determined by the AAM, provides a minimal description of the
  corresponding reaction rule. We have shown formally that it is not
  sufficient to solve a general graph editing problem. Instead, the cyclic
  nature of the transition states must be taken into account as additional
  constraints. With {\altcyc} and {\ilptwo} we have introduced two
  complementary approaches to solve this chemically constrained maximum
  subgraph problem.  The constructive {\altcyc} approach performs better on
  short cycle instances. If more complex transition states need to be
  considered or if flexibility in the cost function is required {\ilptwo}
  becomes the method of choice.

  Advances in high-throughput sequencing technologies drives the
  reconstruction of organism-specific large-scale metabolic networks from
  genomic sequence information. Reaction perception, as illustrated here on
  the Rhea database, is currently the only computational technique to
  suggest missing reactions in the reconstructed networks once the methods
  of comparative genomics to infer enzyme activities are exhausted. We have
  demonstrated here that efficient computation of AAMs serves as first
  effective step. Much remains to be done, however. Most importantly, the
  AAM determines only a minimal reaction rule confined to the reaction
  center. The feasibility of chemical reactions, however, also depends on
  additional context in the vicinity of the reaction center. While graph
  grammar systems readily accommodate non-trivial context
  \cite{Benkoe:03b,Andersen:14a}, we have yet to develop methods to infer the
  necessary contexts from the huge body of chemical reaction
  knowledge. Once this is solved, such more elaborate rules will form a
  highly efficient filter for the candidate AAMs. In this context the
  stepwise construction of the transition state in {\altcyc} holds further
  promise: context information could be used efficiently already in the AAM
  construction step to prune its search tree, simultaneously 
  increase the chemical realism of the solutions and its computational 
  efficiency.

\newpage
\section*{Acknowledgments}
This work was supported in part by the Volkswagen Stiftung proj.\
no.\ I/82719, and the COST-Action CM1304 ``Systems Chemistry'' and
by the Danish Council for Independent Research, Natural Sciences.

\bibliography{atomatom.bib}
\bibliographystyle{plain}

\newpage
\appendix

\section{Statistical Analysis of Rhea}\label{app:rhea}
  Of the $M=3786$ non-isomorphic molecular graphs in
  RHEA, $2204$ are identified uniquely by their sum formula. While
  $2030$ of the molecules appear only in a minimum of 4 reactions,
  some compounds take part in a very large fraction of all reactions
  in RHEA, e.g., \ce{H^+} participates in 11,1147
  reactions, some of which are different descriptions of similar reactions where only the direction of the reaction differs, 5055 of these are truly distinct, adenosine di-, and tri-phosphate (and it derivatives),
  water, and dioxide each participate in more than 2000 reactions
  (depicted as red dots in Fig.~\ref{fig:freq1} (right)).  The maximum number of
  isomers (i.e., compounds that have the same sum formula but a
  non-isomorphic graph representation) is $63$. The corresponding sum
  formula is \ce{C15H24}. Interestingly, most of the large sets of
  isomers in RHEA are terpenes, condensates of identical five carbon
    atom building blocks. The terpenes form a combinatorial class of
    polycyclic ring-systems via complex sequences of cyclisation and
    isomerization reactions.
  Fig.~\ref{fig:freq1} (left) summarizes the results (terpenes marked with red).

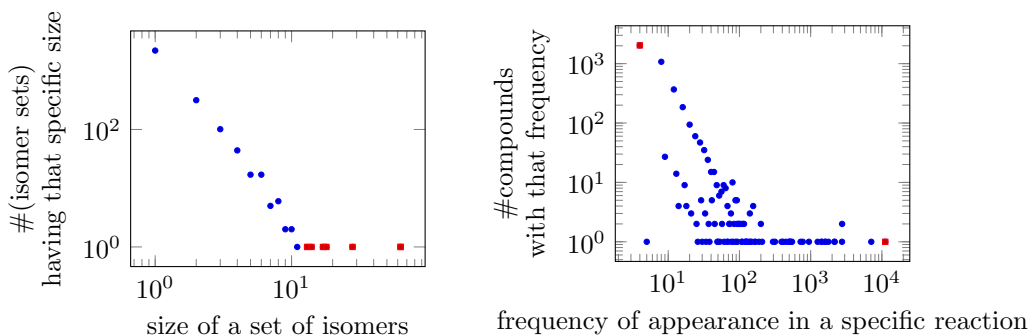
\begin{figure}
\pgfplotsset{width=0.45*\textwidth}
\begin{tikzpicture}
\begin{loglogaxis}[
    xlabel = size of a set of isomers,
    ylabel style={align=center},
    ylabel = \#(isomer sets)\\having that specific size,
    ]
\addplot+[only marks,mark size=1pt]
table {samp2.dat};
\addplot+[only marks,mark size=1pt, mark=square*,red]
table {
x y
13  1 
14  1 
17  1 
18  1 
28 1 
63 1
63 1
};
\end{loglogaxis}
\end{tikzpicture}
\qquad
\begin{tikzpicture}
\begin{loglogaxis}[
    xlabel = frequency of appearance in a specific reaction,
    ylabel style={align=center},
    ylabel = \#compounds\\with that frequency,
    ]
\addplot+[only marks,mark size=1pt]
table {samp1.dat};
\addplot+[only marks,mark size=1pt, mark=square*,red]
table {
x y
4 2030
11147 1
};
\end{loglogaxis}
\end{tikzpicture}
\caption{Distribution of isomers and frequency of participation in reactions in Rhea. Left plot shows a few sets of isomers are very large, while most compounds in Rhea are unique up to sum formula of those compounds. Right plot shows the frequency with which a compound participates in reactions.}
\label{fig:freq1}
\end{figure}

\section{Analysis of Runtime}\label{app:runtime}
As we are mainly interested in single step reactions, we restricted our algorithms to only look for connected, vertex-disjoint transition states during the comparison. Fig.~\ref{fig:solved} shows the fraction of instances where {\altcyc}, {\ilptwo} and a na{\"i}ve ILP-implementation with $O(n^4)$ constraints, \texttt{ILP4}, are able to enumerate all non-equivalent atom-atom mappings for different instance size categories as well as absolute number of instances solved divided by solution size.

Only very few instances that are not completely solved within the first 60 seconds are solved within reasonable time (one hour). So there seems to be a sharp divide between easy and hard instances. From the plot in Fig.~\ref{fig:solved} (left) of the fraction of instances solved fast we observe an exponential decline in ratio of solved instances. This corresponds well with the expected exponential runtime of the algorithms.

\newsavebox{\tempbox}

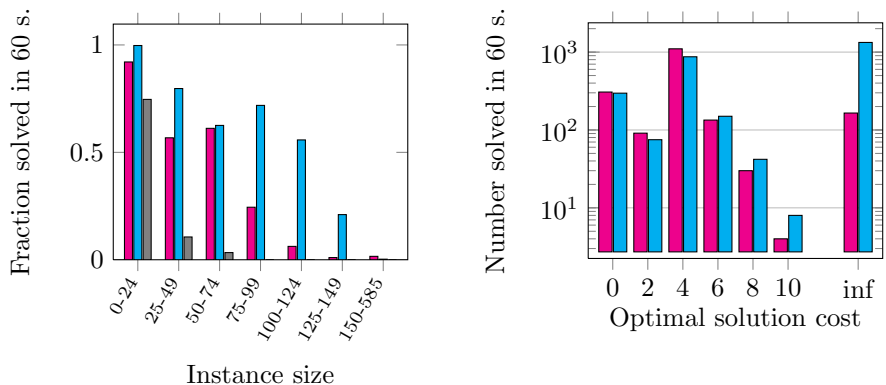
\begin{figure}
\pgfplotsset{width=0.45*\textwidth}
\sbox{\tempbox}{%
\begin{tikzpicture}
\begin{axis}[
    symbolic x coords = {0-24,25-49,50-74,75-99,100-124,125-149,150-585},
    xlabel = Instance size,
    ylabel = Fraction solved in 60 s.,
    ymin = 0,
    ybar=0.5pt,
    x tick label style={rotate=60,anchor=east, font=\scriptsize},
    xtick=data,
    bar width=3,
    xlabel near ticks
    ]
    \addplot[fill=magenta, mark=none] coordinates{
        (100-124, 0.0622222222222)
        (25-49, 0.567660550459)
        (0-24, 0.920698924731)
        (75-99, 0.244521337947)
        (125-149, 0.0101694915254)
        (150-585, 0.015873015873)
        (50-74, 0.611782477341)
    };
    \addplot[fill=cyan, mark=none] coordinates{
        (100-124, 0.557777777778)
        (25-49, 0.797018348624)
        (0-24, 0.997311827957)
        (75-99, 0.718569780854)
        (125-149, 0.210169491525)
        (150-585, 0.0026455026455)
        (50-74, 0.625377643505)
    };
    \addplot[fill=gray, mark=none] coordinates{
        (100-124, 0.0)
        (25-49, 0.10606060606061)
        (0-24, 0.746478873239)
        (75-99, 0.0)
        (125-149, 0.0)
        (150-585, 0.0)
        (50-74, 0.03333333333333)
    };
\end{axis}
\end{tikzpicture}
}
\usebox{\tempbox}
\qquad
\vbox to \ht\tempbox{%
\begin{tikzpicture}
\begin{semilogyaxis}[
    symbolic x coords = {0,1,2,3,4,5,6,7,8,9,10,11,12,13,inf},
    ymin = 0,
    ymajorgrids = true,
    ybar=0.5pt,
    xtick = data,
    xlabel = Optimal solution cost,
    ylabel = Number solved in 60 s.,
    bar width=5
    ]
    \addplot[fill=magenta, mark=none] coordinates {
        (0, 307)
        (2, 91)
        (4, 1102)
        (6, 134)
        (8, 30)
        (10, 4)
        (inf, 165)
    };
    \addplot[fill=cyan, mark=none] coordinates {
        (0, 297)
        (2, 75)
        (4, 871)
        (6, 150)
        (8, 42)
        (10, 8)
        (inf, 1331)
    };
\end{semilogyaxis}
\end{tikzpicture}
\vfil
}
\caption{Fraction and number of instances where all optimal atom-atom maps are found in 60 seconds (user time) by instance size and optimal solution cost for {\altcyc} (magenta), {\ilptwo} (cyan) and \texttt{ILP4} (gray).}
\label{fig:solved}
\end{figure}

As we restricted the solution set certain instances are proven infeasible by {\ilptwo}, while {\altcyc} will continue searching for solutions until the parameter $k$, the number of weight changes, gets arbitrarily high. We chose to deem instances where {\altcyc} found no solutions for $k\leq 10$ infeasible and terminate the search. These two classes of solutions are marked in the rightmost column in Fig.~\ref{fig:solved}. Note that the performance of {\altcyc} on the infeasible class of instances depends heavily on the somewhat arbitrary choice of maximum $k$.

Both ILP models are implemented using CPLEX, an efficient state of the art MIP-solver. {\altcyc} and {\ilptwo} has been tested on a total of 4295 Rhea instances, while \texttt{ILP4} has only been tested on a subset of these of size 250.


\def\deltaw{\texttt{WeightAlongPath}}
\def\complete{\texttt{Complete}}

\section{Algorithmic Details}\label{app:code}
For completeness we include pseudo-code for the sub-procedures used in the paper.

\par\noindent\textbf{Pseudo-code for {\deltaw}:} In {\altcyc$^*$} (see Algorithm \ref{alg:altcyc_nondisjoint}) we need to find all previous changes to an edge $\{i,j\}$ currently under examination, $w_P(\{i,j\})$.

In Algorithm \ref{alg:deltaw} we show how to do this in time $O(|P|)$, where $|P|\in O(k)$. It is possible to find $w_P(e)$ in constant time, but this would require much more complicated data structures or making changes to the graphs we work on and as $k$ is in practice very small, this method is preferred.

To find $w_P(e)$ for a list of paths, add $w_P(e)$ for all paths in the list.

\begin{algorithm}[H]
\caption{\texttt{WeightAlongPath}$(\{i,j\}, P)$}
\label{alg:deltaw}
\begin{algorithmic}
    \STATE $w_P \leftarrow 0$
    \STATE $\sigma \leftarrow 1$
    \FOR{$i'$ \textbf{from} $0$ \textbf{to} $|P|-2$}
        \STATE $j'\leftarrow i'+1$
        \IF{$\{i',j'\}=\{i,j\}$}
            \STATE $w_P \leftarrow w_P + \sigma$
        \ENDIF
        \STATE $\sigma \leftarrow -1\cdot \sigma$
    \ENDFOR
\end{algorithmic}
\end{algorithm}

\par\noindent\textbf{Pseudo-code for {\complete}:} When a transition state candidate $\psi'$ is found we need to ensure it can be extended into a complete atom-atom mapping. This can be done as described in Algorithm \ref{alg:complete}. Note that the two graphs $G_1$ and $G_2$ are assumed implicitly known. The algorithm works both for a single path, $P$, or where $P$ represents a list of paths.

The only non-trivial detail in Algorithm \ref{alg:complete} is that it is not correct to remove all edges in the induced subraph on the domain of $\psi'$, the weight change needs to be sufficient, and there may be unchanged cords to consider.

\begin{algorithm}[H]
\caption{\texttt{Complete}$(\psi', P)$}
\label{alg:complete}
\begin{algorithmic}
    \FOR{$e\in P$}\COMMENT{Here $P$ is considered a set of edges}
        \STATE $w_P \leftarrow {\deltaw}(e, P)$
            \IF{$w_P = w_2(\psi(e)) - w_1(e)$}
                \STATE Remove $e$ from $G_1$ and $\psi(e)$ from $G_2$
            \ELSE
                \STATE \textbf{fail}
            \ENDIF
    \ENDFOR
    \FOR{$(i,p)\in V_1\times V_2$ \textbf{where} $\psi'(i) = p$}
        \STATE Relabel $i$ and $p$ to have identical, otherwise unique labels
    \ENDFOR
    \STATE \textbf{return} an isomorphism from $G_1$ to $G_2$
\end{algorithmic}
\end{algorithm}

\par\noindent\textbf{Finding 2-to-2 Candidates in $O(n^2\log n)$ Comparisons.}\label{app:2to2}
In order to generate all $O(n^4)$ candidate reactions with no more than two molecules in the educts or products we use Algorithm \ref{alg:2to2}. A set of molecules, $M$, is given, as well as a method to obtain the distribution of atoms and charges of the molecules $h$, in practice some implementation of sparse vectors. We assume we keep pointers to the original molecules that resulted in each distribution, and we get these with the function $mol$.
\begin{algorithm}[h!tb]
\caption{\texttt{2to2}$(M)$}
\label{alg:2to2}
\begin{algorithmic}
    \STATE $\mathcal{H} \leftarrow h(M) \cup \{\vec{0}\}$
    \STATE generate $H = \{h_1 + h_2 \mid (h_1, h_2)\in \mathcal{H}\times \mathcal{H} \land h_1 \leq h_2\}$ as an array 
    \STATE \texttt{Sort}$(H)$
    \FOR{$i\leftarrow 1$ \textbf{to} len$(H)-1$}
        \STATE $j\leftarrow i +1$
        \WHILE{$j\leq \text{len}(H) \land H[i] = H[j] $}
            \STATE \textbf{output} $(mol(H[i]), mol(H[j]))$
            \STATE $j\leftarrow j+1$
        \ENDWHILE
    \ENDFOR
\end{algorithmic}
\end{algorithm}

The algorithm is dominated by one of two things, either the sorting of the length $n^2$ array $H$ (where $n = |M|$), or the time to output candidates $k\in O(n^4)$, the resulting runtime is then $O(n^2\log n + k)$.


\section{Example of Aromatic Structure}\label{app:aromatic}

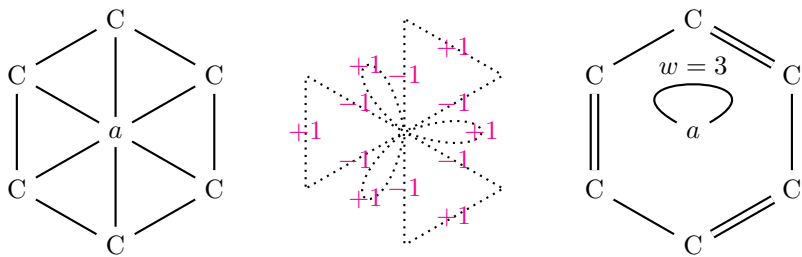
\begin{figure}[h!tb]
\begin{center}
\def\L{1.5}
\def\W{3.8}

\begin{tikzpicture}[every loop/.style={}]
    \node  (c_1_l) at (0*\L, 1*\L)         {C};
    \node  (c_2_l) at (0.866*\L, 0.5*\L)   {C};
    \node  (c_3_l) at (0.866*\L, -0.5*\L)  {C};
    \node  (c_4_l) at (0*\L, -1*\L)        {C};
    \node  (c_5_l) at (-0.866*\L, -0.5*\L) {C};
    \node  (c_6_l) at (-0.866*\L, 0.5*\L)  {C};
    \node  (a_l)   at (0, 0)               {$a$};
    \path
        (c_2_l) edge [bond1] (c_3_l)
        (c_4_l) edge [bond1] (c_5_l)
        (c_5_l) edge [bond1] (c_6_l)
        (c_6_l) edge [bond1] (c_1_l)
        (c_1_l) edge [bond1] (c_2_l)
        (c_3_l) edge [bond1] (c_4_l)
        (c_1_l) edge [bond1] (a_l)
        (c_2_l) edge [bond1] (a_l)
        (c_3_l) edge [bond1] (a_l)
        (c_4_l) edge [bond1] (a_l)
        (c_5_l) edge [bond1] (a_l)
        (c_6_l) edge [bond1] (a_l)
    ;

    \node [dummy] (c_1_c) at (0*\L + \W, 1*\L)         {};
    \node [dummy] (c_2_c) at (0.866*\L + \W, 0.5*\L)   {};
    \node [dummy] (c_3_c) at (0.866*\L + \W, -0.5*\L)  {};
    \node [dummy] (c_4_c) at (0*\L + \W, -1*\L)        {};
    \node [dummy] (c_5_c) at (-0.866*\L + \W, -0.5*\L) {};
    \node [dummy] (c_6_c) at (-0.866*\L + \W, 0.5*\L)  {};
    \node [dummy] (a_c) at (0 + \W, 0) {};
    \path
        (c_1_c) edge [change] node {$+1$} (c_2_c)
        (c_3_c) edge [change] node {$+1$} (c_4_c)
        (c_5_c) edge [change] node {$+1$} (c_6_c)
        (c_4_c) edge [change] node {$-1$} (a_c)
        (c_2_c) edge [change] node {$-1$} (a_c)
        (c_6_c) edge [change] node {$-1$} (a_c)
        (c_1_c) edge [change] node {$-1$} (a_c)
        (c_3_c) edge [change] node {$-1$} (a_c)
        (c_5_c) edge [change] node {$-1$} (a_c)
        (a_c) edge [change, in=-20, out=20, loop, min distance=40] node {$+1$} (a_c)
        (a_c) edge [change, in=120-20, out=120+20, loop, min distance=40] node {$+1$} (a_c)
        (a_c) edge [change, in=240-20, out=240+20, loop, min distance=40] node {$+1$} (a_c)
        ;

    \node  (c_1_r) at (0*\L + 2*\W, 1*\L)         {C};
    \node  (c_2_r) at (0.866*\L + 2*\W, 0.5*\L)   {C};
    \node  (c_3_r) at (0.866*\L + 2*\W, -0.5*\L)  {C};
    \node  (c_4_r) at (0*\L + 2*\W, -1*\L)        {C};
    \node  (c_5_r) at (-0.866*\L + 2*\W, -0.5*\L) {C};
    \node  (c_6_r) at (-0.866*\L + 2*\W, 0.5*\L)  {C};
    \node  (a_r)   at (0 + 2*\W, 0)               {$a$};
    \path
        (c_1_r) edge [bond2] (c_2_r)
        (c_2_r) edge [bond1] (c_3_r)
        (c_3_r) edge [bond2] (c_4_r)
        (c_4_r) edge [bond1] (c_5_r)
        (c_5_r) edge [bond2] (c_6_r)
        (c_6_r) edge [bond1] (c_1_r)
        (a_r) edge [bond1, in=150, out=30, loop] node [above] {$w=3$} (a_r);
        
\end{tikzpicture}

\end{center}
\caption{Illustration of modeling of aromatic cycle. Left is an aromatic cycle, right is the same cycle in Kukel\'e form. Edge labels are not shown, and edge weight is implied with multiple parallel lines. The figure in the middle depicts the alternating transition state between the two assuming AAM by position of atoms.}
\label{fig:aromatic_ex}
\end{figure}

It is non-trivial that the model presented here of aromatic complexes will allow for AAMs with cyclic transitions states, Fig.~\ref{fig:aromatic_ex} illustrates how this can be done. We add special aromatic vertices with loops to either $G_1$ or $G_2$ to ensure the AAM is still a bijection and that a mapping is feasible.

\end{document}